\documentclass[12pt]{article}
\usepackage{a4,amsthm,amsmath,amsfonts,amssymb,cite}
\usepackage{booktabs}
\usepackage{hyperref}

\usepackage{amsfonts,graphicx,epsfig,epstopdf,algorithmic}
\ifpdf
  \DeclareGraphicsExtensions{.eps,.pdf,.png,.jpg}
\else
  \DeclareGraphicsExtensions{.eps}
\fi

\usepackage{tikz}
\usetikzlibrary{shapes, arrows, automata, positioning, fit, shadows, decorations.pathmorphing}
\usetikzlibrary{calc, backgrounds}

\usepackage{geometry}
\usepackage{enumitem}
\usepackage{tabularx,xspace,booktabs}

\newtheorem{theorem}{Theorem}[section]
\newtheorem{lemma}[theorem]{Lemma}

\newtheorem{corollary}[theorem]{Corollary}
\newtheorem{definition}[theorem]{Definition}

\providecommand{\mIs}[1]{\mI_{(#1)}}
\providecommand{\mnk}[1]{\mN_{(#1)}}
\providecommand{\mmk}[2]{\matbld{#1}_{(#2)}}
\providecommand{\vvk}[2]{\vecn{#1}_{(#2)}}
\providecommand{\vvd}[1]{\vecn{d}^{(#1)}}

\providecommand{\subscore}[4]{\textrm{#1}(#2,#3,#4)}

\usepackage{xcolor}

\newcommand{\gentropy}[1]{\ensuremath{I_{f}(G, #1)}}

\DeclareMathOperator*{\avg}{avg}

\newcommand{\pscf}[2]{\ensuremath{p_{f}(#1,#2)}}

\newcommand{\ourparag}[1]{\textbf{#1}.}

\usepackage[utf8]{inputenc}
\usepackage{amsfonts,amsmath,amssymb,array}

\newcommand{\eps}{\varepsilon}

\DeclareMathOperator{\diag}{diag}
\DeclareMathOperator{\tr}{tr}
\providecommand{\vecn}[1]{\ensuremath{\textbf{#1}}}

\providecommand{\vb}{\vecn{b}}

\providecommand{\ve}{\vecn{e}}

\providecommand{\vn}{\vecn{n}}

\providecommand{\vu}{\vecn{u}}
\providecommand{\vv}{\vecn{v}}

\providecommand{\vx}{\vecn{x}}
\providecommand{\vy}{\vecn{y}}

\newcommand{\bmat}[1]{ \begin{bmatrix} #1 \end{bmatrix} }

\providecommand{\matbld}[1]{\ensuremath{\textbf{#1}}}
\providecommand{\mA}{\matbld{A}}
\providecommand{\mB}{\matbld{B}}

\providecommand{\mI}{\matbld{I}}

\providecommand{\mM}{\matbld{M}}
\providecommand{\mN}{\matbld{N}}

\providecommand{\mV}{\matbld{V}}
\providecommand{\mW}{\matbld{W}}

\graphicspath{{.}{./figures/}}

\usepackage{csquotes}

\title{Subgraph centrality and walk-regularity\thanks{The work of these authors was supported in part by the Gordon \& Betty Moore Foundation's Data-Driven Discovery Initiative through Grant GBMF4560 to Blair D.~Sullivan.}}

\author{
    Eric Horton
    \thanks{Department of Computer Science, North Carolina State University, Raleigh, NC 27695, USA. Email:
    \texttt{ewhorton,kakloste,blair\_sullivan@ncsu.edu}.
}
\and Kyle Kloster\footnotemark[2]
\and Blair D. Sullivan\footnotemark[2]}

\usepackage{amsopn}

\begin{document}

\maketitle

% REQUIRED
\begin{abstract}
  Matrix-based centrality measures have enjoyed significant popularity in network analysis, in
no small part due to our ability to rigorously analyze their behavior as parameters vary.
Recent work has considered the relationship between subgraph centrality,
which is defined using the matrix exponential $f(x) = \exp(x)$, and the walk structure of a network.
In a walk-regular graph, the number of closed walks of each length must be the same for all nodes, implying uniform $f$-subgraph centralities for any $f$ (or maximum $f$-\emph{walk entropy}).
We consider when non--walk-regular graphs can achieve maximum entropy, calling such graphs \emph{entropic}.
For parameterized measures, we are also interested in which values of the parameter witness this
uniformity.
To date, only one entropic graph has been identified, with only two witnessing parameter values, raising the question of how many such graphs and parameters exist.
We resolve these questions by constructing infinite families of entropic graphs, as well as a family of witnessing parameters with a limit point at zero.\newline

%  05C50, % Graphs and linear algebra (matrices, eigenvalues, etc.)
%  05C75, % Structural characterization of families of graphs
%  15A16 % Matrix exponential and similar functions of matrices

\end{abstract}

\section{Introduction}\label{sec:intro}
Evaluating the relative importance of nodes in a graph is a fundamental operation in network analysis, and the literature is full of well-studied approaches for quantifying node importance (typically called centrality measures)~\cite{estrada2012structure,newman2010networks}.
Functions of matrices are a natural candidate for such rankings~\cite{estrada2010network}; in particular the matrix resolvent~\cite{gleich2015pagerank,katz1953new,page1999pagerank} (Katz and PageRank centrality) and the matrix exponential~\cite{benzi2013total,benzi2014matrix,chung2007heat,estrada2005subgraph} (heat kernel and subgraph centrality) have been widely studied and used in practice.

Like many approaches to centrality, functions of matrices actually give rise to infinite families of specific centrality measures based on the value of some parameter in the definition.
For example, the ($\beta$-)\emph{subgraph centrality} of a node $i$ is given by the diagonal entry of the matrix exponential $\exp(\beta \mA)_{ii}$ for non-negative $\beta$~\cite{estrada2005subgraph}.
In general, there is no consensus on the ``best'' parameter value(s) for a centrality measure, and the effects of specific choices can be hard to characterize.
Recent work has considered how relative node rankings change as matrix-based centrality
parameters vary~\cite{benzi2015limiting,paton2017centrality}.
In this work, we study when subgraph centralities can assign identical scores to nodes that are structurally different in the underlying graph.

The notion of ``structural equivalence'' we consider is based on prior studies of the interplay between the uniformity of subgraph centrality (or \emph{walk entropy}~\cite{benzi2014note,estrada2014walk}) and the \emph{walk-regularity} of a graph.
To be more precise, consider a graph $G$ with adjacency matrix $\mA$.
We say $G$ is \emph{walk-regular} if for each $\ell \geq 0$, every node has the same number of closed walks of length $\ell$ (equivalently, $\mA^{\ell}$ has constant diagonal), and we say nodes $i$ and $j$ are in the same \emph{walk class} (structurally equivalent) if $\mA^{\ell}_{ii} = \mA^{\ell}_{jj}$.
Early studies suggested that walk-regularity might be completely characterized by attaining maximum walk entropy;
that is, it was conjectured that a graph is walk-regular if and only if there exists at least one $\beta_0$ such that $\exp(\beta_0 \mA)$ gives all nodes the same score (resulting in maximum walk entropy)~\cite{benzi2014note}.
Recent work~\cite{Kloster2018115} disrupted this line of research by presenting a single graph which is non--walk-regular yet has uniform $\beta$-subgraph centrality for a particular value of $\beta$.
We call such a non--walk-regular graph \emph{entropic}, and any value $\beta$ for which $G$ attains maximum walk entropy an \emph{entropic value} for $G$.

In this work, we resolve several outstanding questions regarding the interplay between
walk-regularity and centrality. We begin by exhibiting an infinite family of entropic graphs (Section~\ref{sec:cartesian-graph}).
Our construction proves that for each entropic value $\beta_0$ there are infinitely many graphs entropic with respect to $\beta_0$.
Interestingly, this result does not produce any \emph{new} entropic values $\beta_0$;
however, in Section~\ref{sec:beta-distribution}
we establish that the set of entropic values is not only infinite, but contains a limit point at 0.

We then consider the more general class of $f$-\emph{subgraph centralities} given by the
diagonal entries of $f(\beta\mA)$ for a parameter $\beta > 0$ and suitable function $f$ defined on the spectrum of $\mA$.
If a graph has uniform $f$-subgraph centrality for some parameter $\beta$,
we say $G$ is \emph{$f$-entropic}, the value $\beta$ is $f$-entropic with respect to $G$, and $f$ is an entropic function.
In Section~\ref{sec:infinite-tensor} we prove that there are infinitely many functions $f_i$ that are entropic with respect to at least one graph.

Finally, we consider the related question of when a subset of a graph's walk-classes
have the same $f$-subgraph centrality score for some parameter $\beta$;
when this occurs, we say the walk classes \emph{collide} at $f$,
or that $f$ \emph{induces a collision} at the walk-classes.
Note that a graph is $f$-entropic exactly when $f$ induces a collision at all of the graph's walk-classes.
We present a sufficient condition for determining that a set of walk-classes collide under some function $f$,
and a sufficient condition for concluding that a set of walk-classes do not collide under \emph{any} suitable function $f$.
These sufficient conditions are practical to compute on modest-sized graphs.

\section{Background}\label{sec:prelims}\label{sec:notation}
We consider only simple, loopless, unweighted, undirected graphs.
Given a graph $G = (V,E)$, we label its nodes as $1, 2, ..., n = |V|$ and denote its adjacency matrix by $\mA_{G}$, or $\mA$ when context is clear.
$K_c$ denotes the complete graph on $c$ nodes.
Uppercase bold letters indicate matrices, $\mI$ is reserved for the square identity matrix, and
we write $\mI_{(n)}$ to indicate dimensions $n\times n$ if they are not clear from context.
The entry in the $i$th row and $j$th column of $\mA$ is denoted $\mA_{ij}$.
The $j$th column of $\mA$ is $\mA(:,j)$, and the subvector of $\mA(:,j)$
with row indices in set $S$ is $\mA(S,j)$.

Vectors are denoted by lowercase bold letters, e.g. $\vv$, and we write $\vvk{v}{\ell}$ to indicate length $\ell$. We write $\ve_j$ for the $j$th column of the identity matrix and $\ve$ for the vector of all 1s. To refer to entry $i$ of a vector we write $\vv(i)$, or $\vv_i$ if it is not ambiguous.

In the rest of this section we provide background on walk-regularity and walk-classes, graph entropy and centrality measures, as well as our definitions for entropic graphs, functions, and parameters

\ourparag{Walk-classes and walk-regularity}
We call a walk of length $\ell$ an \emph{$\ell$-walk}.
A graph $G$ is \emph{walk-regular} if and only if for each $\ell \geq 0$, each node in $G$ is incident to the same number of closed $\ell$-walks.
Two nodes $u, v$ are in the same \emph{walk-class}
if and only if for every $\ell \geq 0$ they are incident to the same number of closed $\ell$-walks.
This is equivalent to having $(\mA^\ell)_{uu} = (\mA^\ell)_{vv}$ for $\ell \geq 0$;
by the Cayley-Hamilton theorem, it suffices to consider only the $m$ values $\ell = 0, \cdots, m-1$ where $m\leq n$ is the degree of the minimal polynomial of $\mA$.
We remark that $G$ is walk-regular if and only if it has exactly one walk-class.

To analyze the walk-classes of a graph $G$ with adjacency matrix $\mA$, we define the \emph{walk matrix} as follows.
First, define the vector $\vvd{\ell}_{\mA}$ entrywise by $\vvd{\ell}_{\mA,j} = (\mA^\ell)_{jj}$ for each $j=1,\cdots,n$ and every $\ell \geq 0$.
When the context is clear, we use $\vvd{\ell} = \vvd{\ell}_{\mA}$.
Because we consider only loopless graphs, $\mA^1$ is constant diagonal, so the term $\vvd{1}$ provides no distinguishing information about the nodes.
The walk matrix $\mW_{\mA}$ is then
$\mW_{\mA} = \bmat{ \vvd{2} & \cdots & \vvd{n-1} }$, which we denote by $\mW$ if it is not ambiguous.

\noindent Given a square matrix $\mM \in \mathbb{R}^{n \times n}$, we say $\mM$ is \emph{constant diagonal} if every entry $\mM_{jj}$ is the same value.
A graph $G$ is walk-regular if and only if for each integer $\ell \geq 0$, the matrix $\mA^{\ell}$ is constant diagonal~\cite{godsil2013algebraic}.

\ourparag{Centrality, entropy, and functions of matrices}
If a function $f(x)$ with power series $\sum_{k=0}^{\infty} c_k x^k$ is defined on the spectrum of $\mA$, we can express
\[
    f(\mA) = \sum_{k=0}^{\infty} c_k \mA^k.
\]
Note that all nodes in the same walk-class will have the same diagonal value $f(\mA)_{ii}$,
and for a walk-regular graph, $f(\mA)$ is constant diagonal for any $f$ defined on $\mA$.
For more detail on the conditions $f(x)$ must satisfy to be defined on $\mA$, see~\cite{higham2008functions}.

For any function $f$ defined on the spectrum of $\mA$, the $f$-\emph{subgraph centrality} of node $j$ is given by $f(\mA)_{jj}$.
For $f(x) = \exp(x)$, this is simply called \emph{subgraph centrality}~\cite{estrada2005subgraph}.
An $f$-subgraph centrality and parameter $\beta$ define a probability distribution $p_f(\beta)$ on $V(G)$ by normalizing $\pscf{j}{\beta} = f(\beta\mA)_{jj} / \tr( f(\beta\mA) )$.

Dehmer introduced the concept of \emph{graph entropy} to study uniformity of various graph structures~\cite{dehmer2008information}.
Given any probability distribution $p:V(G)\rightarrow [0,1]$ on the nodes of a graph $G$, the corresponding \emph{graph entropy} is defined by $I_p(G) = - \sum_{j=1}^n \left( p(j) \cdot \log p(j) \right)$.
Graph entropies take values in $[0, \log n]$ and attain $\log n$ if and only if the distribution $p$ is uniform.
Thus, $G$ has maximum graph entropy with respect to $p_f$ exactly when $f(\mA)$ is constant diagonal.

For a fixed function $f(x)$, we study a family of associated graph entropies given by $f(\beta x)$
for varying $\beta$.
We use $\pscf{j}{\beta}$ to denote the probability distribution that arises from the centrality values $f(\beta \mA)_{jj}$, and we use $\gentropy{\beta}$ to denote the corresponding graph entropy.

The \emph{walk entropy} is defined in terms of subgraph centrality for a parameter $\beta$:
\begin{align}\label{eqn:walk-entropy}
  S^V(G, \beta) = -\sum_{j=1}^n \left( \frac{\exp(\beta\mA)_{jj}}{\tr (\exp(\beta\mA))}  \log \frac{\exp(\beta\mA)_{jj}}{\tr (\exp(\beta\mA))} \right).
\end{align}
A closer analysis of walk entropy was initiated in~\cite{estrada2014walk}, where it was conjectured that a graph is walk-regular if and only if its walk entropy is maximized for all $\beta \geq 0$.
A stronger form of this conjecture was proven in~\cite{benzi2014note}, namely that walk-regularity follows if walk entropy is maximized for all $\beta \in I$, for any set $I\subset \mathbb{R}$ with a limit point.
It was further conjectured that walk-regularity follows if there exists even a single value $\beta > 0$ such that walk entropy is maximized, but~\cite{Kloster2018115} exhibited a counterexample which we refer to as $G(4,5)$;
we introduce this notation in Section~\ref{sec:beta-distribution}.

\ourparag{Entropic graphs, functions, and values}
In this paper we restrict our attention to functions $f$ that have power-series representations with coefficients that are all positive.
Additionally, we assume that the power series has positive radius of convergence.
The authors in~\cite{benzi2015limiting} explored this exact setting;
previous work had considered a slight variant, allowing some coefficients to be nonnegative~\cite{estrada2010network, rodriguez2007functional}.
\begin{definition}\label{def:positive-power}
  A function $f(x)$ is a \emph{positive power-series coefficient} (PPSC) function if it has a power series $f(x) = \sum_{k=0}^{\infty} c_k x^k$ with $c_k > 0$ $\forall k$.
\end{definition}
This class of functions includes functions associated with several popular centrality measures.
In particular, the matrix resolvent $f(\beta \mA) = (\mI - \beta\mA)^{-1}$~(Katz centrality~\cite{katz1953new} and PageRank~\cite{gleich2015pagerank,page1999pagerank})
and the matrix exponential $f(\beta x) = \exp(\beta \mA)$~(subgraph centrality~\cite{estrada2005subgraph}, total subgraph communicability~\cite{benzi2013total}, and heat kernel centrality~\cite{chung2007heat}) have been widely studied.

\begin{definition}\label{def:entropic}
  A graph $G$ is \emph{$f$-entropic} if $G$ is connected and non--walk-regular, and
  there exists a PPSC function, $f$, such that $f(\beta_0\mA)$ is constant diagonal for some $\beta_0 > 0$.
  We say $G$ is \emph{$f$-entropic with respect to $\beta_0$}, $f$ is \emph{entropic on $G$}, and we call $\beta_0$ an \emph{$f$-entropic value}.
  If $f(x) = \exp(x)$, we simply say $G$ is \emph{entropic} and $\beta_0$ is an \emph{entropic value}.
\end{definition}

Conversely, we call a value $\beta_0$ for which no entropic graph exists \emph{sub-entropic}.
That is, if $\beta_0$ is sub-entropic, then for \emph{any} graph's adjacency matrix $\mA$, $\exp(\beta_0 \mA)$ is constant diagonal if and only if the graph is walk-regular.
When using $f(x)\neq \exp(x)$, we say $\beta_0$ is $f$-sub-entropic.
Pursuit of a concrete, sub-entropic value $\beta_0$ (and in particular the conjectured value $\beta_0 = 1$) has motivated much of the recent literature on the topic.
See Conjecture 3 in~\cite{estrada2013discriminant} and Conjecture 3.1 in~\cite{benzi2014note}.

\section{An infinite family of entropic graphs}\label{sec:cartesian-graph}
Here we prove that there are infinitely many entropic graphs.
We construct the graphs via the Cartesian product.
The \emph{Cartesian product} of two graphs, $G$ and $H$, is denoted $G \Box H$ and satisfies the algebraic identity $\mA_{G \Box H} = \mA_G \otimes \mI + \mI \otimes \mA_H.$
\begin{lemma}\label{lem:cartesian-const-diag}
    Let $G, H$ be graphs and $\beta_0 > 0$ a value such that $\exp( \beta_0 \mA_{G})$ and $\exp( \beta_0 \mA_{H})$ are constant diagonal. Then $\exp(\beta_0 \mA_{G\Box H})$ is constant diagonal.
\end{lemma}
\begin{proof}
  Recall that for any matrices $\mB_1$ and $\mB_2$ that commute, we have $\exp(\mB_1 + \mB_2) = \exp(\mB_1)\exp(\mB_2)$.
  This allows us to write
  \[
    \exp( \beta \mA_{G \Box H}) \hspace{4pt} = \hspace{4pt} \exp( \beta\mA_G \otimes \mI  +  \beta\mI \otimes \mA_{H}) \hspace{4pt} = \hspace{4pt} \exp(\beta \mA_G)\otimes \exp(\beta \mA_H).
  \]
  The second equality follows because $\exp(\mB \otimes \mI) = \exp(\mB)\otimes \mI$ holds for any square $\mB$.
  Thus, if $\exp(\beta \mA_G)$ and $\exp(\beta \mA_H)$ have constant diagonal, so does $\exp( \mA_{G \Box H})$.
\end{proof}
\begin{lemma}\label{lem:cart-non-walk-reg}
  For graphs $G, H$ not both walk-regular, $G\Box H$ is not walk-regular.
\end{lemma}
\begin{proof}
    Without loss of generality, let $G$ be not walk-regular.
    Then by a result of Benzi (~\cite{benzi2014note},~Theorem~2.1), there exists some $\beta_0 > 0$ such that $\exp(\beta_0 \mA_G)$ is not constant diagonal.
    Thus, for any graph $H$, we have that $\exp(\beta_0 \mA_G)\otimes\exp(\beta_0 \mA_H)$ is not constant diagonal, which implies $\exp(\beta_0 \mA_{G\Box H})$ is not constant diagonal, so $G \Box H$ is not walk-regular.
\end{proof}
\begin{theorem}\label{thm:cartesian-entropic}
    Let $G, H$ be graphs and $\beta_0 > 0$ such that $\exp( \beta_0 \mA_{G})$ and $\exp( \beta_0 \mA_{H})$ are constant diagonal.
    If $G$ and $H$ are connected and at least one is entropic, then $G \Box H$ is entropic with respect to $\beta_0$.
\end{theorem}
\begin{proof}
  Since we assume $\exp(\beta_0 \mA_{G})$ and $\exp(\beta_0 \mA_H)$ are both constant diagonal for some $\beta_0 > 0$, then by Lemma~\ref{lem:cartesian-const-diag}, $\exp(\beta_0 \mA_{G\Box H})$ is constant diagonal too. If we show $G \Box H$ is connected and not walk-regular, then the result follows.

  By assumption, at least one of $G, H$ is entropic and therefore not walk-regular, so by Lemma~\ref{lem:cart-non-walk-reg}, $G \Box H$ is not walk-regular.
  Finally, a result of Chiue and Shieh~(~\cite{chiue1999connectivity}, Lemma 3) implies that $G \Box H$ is connected if and only if both $G$ and $H$ are connected, so we know $G \Box H$ is connected.
\end{proof}
\begin{corollary}\label{cor:cartesian-infinite}
    Let $G$ be a graph entropic with respect to $\beta_0 > 0$, and let $H$ be any connected, walk-regular graph. Then any finite product $H \Box (H \Box ( \cdots (H \Box G)))$ is entropic with respect to $\beta_0$.
\end{corollary}
In particular, we observe that the entropic graph $G(4,5)$ gives an infinite family of graphs entropic with respect to two values of $\beta$, for any connected, walk-regular graph $H$.
We remark that all graphs in this family are entropic with respect to the same two values $\beta$ for which $G(4,5)$ is entropic.

\section{The distribution of entropic values}\label{sec:beta-distribution}

We now consider the set of all entropic values.
By a result of Kloster et al.~(~\cite{Kloster2018115}, Theorem~3), this set is at most countable; one important question is whether it is actually finite.
Although Corollary~\ref{cor:cartesian-infinite} exhibits an infinite family of entropic graphs, they all share the same two entropic values. The distribution of entropic values is also of interest; in particular, we ask whether there are intervals where no values are entropic; we say an interval $(a,b) \subset \mathbb{R}$ is a \emph{sub-entropic interval} if no $\beta \in (a,b)$ is entropic for any graph.
Since subgraph centrality rankings converge to degree-rankings as $\beta$ converges to 0~\cite{benzi2015limiting}, walk-classes with different node degrees are distinguishable;
this naturally leads to the question of whether there is a sub-entropic interval near zero.

In this section, we construct an infinite sequence of entropic $\beta$ that has 0 as a limit point,
showing that no interval of the form $(0, \eps)$ is sub-entropic;
This result is the first to show that there are infinitely many entropic values, although it leaves open the question of whether the set of entropic values is dense anywhere on the number line, and whether a sub-entropic interval might exist away from 0.

\subsection{Entropic values accumulate at zero}\label{sec:infinite-betas}
Here we prove that no $\eps > 0$ exists such that $(0,\eps)$ is a sub-entropic interval for $\exp(\beta x)$.
We proceed by constructing a sequence of entropic values $\beta_j$ that converges to 0.
Our sequence depends on a family of graphs which generalizes the 24-node graph in~\cite{Kloster2018115}.
In the remainder of this section, we describe the graph family, derive its eigendecomposition, and use the eigendecomposition to construct the sequence of entropic values.

\subsubsection*{The graph class $G(c,m)$}
We define the graph $G(c,m)$ as follows. Given $m$ cliques of size $c$ and an independent set of size $c$, create a perfect matching from the independent set to each clique.
This results in a graph with $c(m+1)$ nodes and two distinct walk-classes---one formed by the nodes in the independent set, the other by the nodes in the cliques.
The 24-node graph in~\cite{Kloster2018115} is $G(4,5)$.
See Figure~\ref{fig:generalized-kral} for a visualization of $G(c,m)$ and its adjacency matrix.

\begin{figure}
      \centering
      \resizebox{0.7\linewidth}{!}{\definecolor{mygrey}{RGB}{180, 180, 180}
\definecolor{myblue}{RGB}{40, 40, 200}

%% Thank you stack exchange !
% https://tex.stackexchange.com/questions/18389/tikz-node-at-same-x-coordinate-as-another-node-but-specified-y-coordinate
%
\makeatletter
\newcommand{\gettikzxy}[3]{%
  \tikz@scan@one@point\pgfutil@firstofone#1\relax
  \edef#2{\the\pgf@x}%
  \edef#3{\the\pgf@y}%
}
\makeatother

\begin{tikzpicture}

% Styles
\tikzstyle{node_style} = [state, fill=white, drop shadow, minimum width=0.05cm]
\tikzstyle{ind_node_style} = [state, fill=mygrey, drop shadow, minimum width=0.05cm]
\tikzstyle{invisi_node_style} = [state, draw=white, fill=white, minimum width=0.05cm]
\tikzstyle{outside_node_style} = [state, fill=white, drop shadow, minimum width=0.05cm, dashed]

\tikzstyle{label_node_style} = [fill=white]

\tikzstyle{edge_style} = [line width=2pt]
\tikzstyle{path_style} = [line width=2pt, decorate, decoration={snake}]
\tikzstyle{arrow_style} = [draw, fill=white, single arrow, single arrow head indent=1ex, minimum size=1cm, drop shadow]

\def\labelheight{0.8}
\def\distancebetweennodes{ sqrt(\xdist*\xdist + \ydist*\ydist) }
\def\xdist{2.0}
\def\ydist{0.0}

\def\ellipsisangle{atan(\ydist/\xdist)}

% Independent set
\node[ind_node_style] (b1) {$\textbf{1}$};
\node[ind_node_style] (b2) at ([xshift=\distancebetweennodes*1cm]b1) {$\textbf{2}$};
\node[rotate=\ellipsisangle] (bdots) at ([xshift=\distancebetweennodes*1cm]b2) {{\huge $\cdots$}};
\node[ind_node_style] (b3) at ([xshift=\distancebetweennodes*1cm]bdots) {$\textbf{c}$};

% Clique 2
\node (c21) [node_style, above left = 3cm and 1.5cm of b1] {$\textbf{1}$};
\node[node_style] (c22) at ([xshift=\distancebetweennodes*1cm]c21) {$\textbf{2}$};
\node[rotate=\ellipsisangle] (c2dots) at ([xshift=\distancebetweennodes*1cm]c22) {{\huge $\cdots$}};
\node[node_style] (c23) at ([xshift=\distancebetweennodes*1cm]c2dots) {$\textbf{c}$};
% Clique 2 label
\node (c2dummy) at ($(c22)!0.5!(c2dots)$) {};
\node (label2) [label_node_style, above = \labelheight*1cm of c2dummy] {\textcolor{myblue}{{\Large $C_2$}}};

% Clique 1
\node[node_style] (c13) at ([xshift=-\distancebetweennodes*1.5cm]c21) {$\textbf{c}$};
\node[invisi_node_style] (c1dots) at ([shift=({-\xdist*1cm,-\ydist*1cm})]c13) {};
\node[rotate=\ellipsisangle] at (c1dots) {{\huge $\cdots$}};
\node[node_style] (c12) at ([shift=({-\xdist*1cm,-\ydist*1cm})]c1dots) {$\textbf{2}$};
\node[node_style] (c11) at ([shift=({-\xdist*1cm,-\ydist*1cm})]c12) {$\textbf{1}$};
% Clique 1 label
% \gettikzxy{(c11)}{\cax}{\cay}
% \gettikzxy{(c13)}{\cbx}{\cby}
% \node[label_node_style] (label2) at (\cax,\cby) {\textcolor{myblue}{{\Large $C_1$}}};
%
\node (c1dummy) at ($(c12)!0.5!(c1dots)$) {};
\node (label1) [label_node_style, above = \labelheight*1cm of c1dummy] {\textcolor{myblue}{{\Large $C_1$}}};

% Clique m
\node[node_style] (c31) at ([xshift=\distancebetweennodes*4.4cm]c23) {$\textbf{1}$};
\node[node_style] (c32) at ([shift=({\xdist*1cm,-\ydist*1cm})]c31) {$\textbf{2}$};
\node[invisi_node_style] (c3dots) at ([shift=({\xdist*1cm,-\ydist*1cm})]c32) {};
\node[rotate=-\ellipsisangle] at (c3dots) {{\huge $\cdots$}};
\node[node_style] (c33) at ([shift=({\xdist*1cm,-\ydist*1cm})]c3dots) {$\textbf{c}$};
% Clique m label
% \gettikzxy{(c31)}{\cax}{\cay}
% \gettikzxy{(c33)}{\cbx}{\cby}
% \node[label_node_style] (label3) at (\cbx,\cay) {\textcolor{myblue}{{\Large $C_m$}}};
\node (c3dummy) at ($(c32)!0.5!(c3dots)$) {};
\node (label3) [label_node_style, above = \labelheight*1cm of c3dummy] {\textcolor{myblue}{{\Large $C_m$}}};

% Ellipsis between cliques
\node (ellipses) at ($(c23)!0.5!(c31)$) {{\Huge $\cdots$}};

%%
%% EDGES
%%

% Edges from Independent set to Clique 1
\draw (b1) edge[edge_style] (c11);
\draw (b2) edge[edge_style] (c12);
\draw (b3) edge[edge_style] (c13);

% Edges from Independent set to Clique 2
\draw (b1) edge[edge_style] (c21);
\draw (b2) edge[edge_style] (c22);
\draw (b3) edge[edge_style] (c23);

% Edges from Independent set to Clique 2
\draw (b1) edge[edge_style] (c31);
\draw (b2) edge[edge_style] (c32);
\draw (b3) edge[edge_style] (c33);

% CLIQUE EDGES
\def\edgebend{30}

% Clique 1
\draw (c11) edge[edge_style] (c12);
\draw (c11) edge[edge_style, bend left = \edgebend] (c13);
\draw (c12) edge[edge_style, bend left = \edgebend] (c13);

% Clique 2
\draw (c21) edge[edge_style] (c22);
\draw (c21) edge[edge_style, bend left = \edgebend] (c23);
\draw (c22) edge[edge_style, bend left = \edgebend] (c23);

% Clique 1
\draw (c31) edge[edge_style] (c32);
\draw (c31) edge[edge_style, bend left = \edgebend] (c33);
\draw (c32) edge[edge_style, bend left = \edgebend] (c33);

\end{tikzpicture}}%
      \vspace{10pt}
      \scalebox{0.85}{
          $ \mA_{G(c,m)} = \left(
                                  \begin{array}{cccccc}
                                      0 & \mIs{c} & \hdots &  \mIs{c} \\
                                      \mIs{c} & \mA_{K_c} & 0 & 0 \\
                                      \vdots & 0 & \ddots & 0  \\
                                      \mIs{c} & 0 & 0 & \mA_{K_c} \\
                                  \end{array}
                            \right)
                        = \left(
                            \begin{array}{cc}
                                0 & \vvk{e}{m}^T \otimes \mIs{c} \\
                                \vvk{e}{m} \otimes \mIs{c} & \mIs{m} \otimes \mA_{K_c}
                            \end{array}
                          \right)$
      }%
    \caption{\label{fig:generalized-kral}
        \emph{(Top)} The graph $G(c,m)$ is composed of $m$ copies of the complete graph $K_c$ (indicated by the labels $C_j$). Each clique is connected by a perfect matching to the independent set of size $c$ indicated by the gray nodes at the bottom.
        \emph{(Bottom)} The adjacency matrix for the graph $G(c,m)$. Expressing $\mA_{G(c,m)}$ in terms of Kronecker products with $\ve$ and $\mA_{K_c}$ makes it easier to verify that the vectors in Table~\ref{tab:eigenvectors} are eigenvectors of $\mA_{G(c,m)}$.
    }
\end{figure}

The proof of our main result relies on having an explicit eigendecomposition $(\lambda_k, \vv_k)$ and applying the identity
\[
    \exp(\beta \mA) = \sum_{k=1}^n \exp( \beta \lambda_k ) \vv_k \vv_k^T,
\]
so next we derive an eigendecomposition for $\mA_{G(c,m)}$.
As a first step, it will be convenient to look at the eigendecomposition of a related matrix.

\paragraph{Eigendecomposition of a clique}
Since $\mA_{K_c} = \vvk{e}{c}\vvk{e}{c}^T - \mIs{c}$,
$\mA_{K_c}$ has the same eigendecomposition as $\vvk{e}{c}\vvk{e}{c}^T$ but with eigenvalues shifted by $-1$.
Next we explicitly derive an orthonormal basis for the nullspace of $\ve\ve^T$.

Let $\mmk{H}{c}$ be the $c \times c$ Householder reflector matrix that maps the vector $\ve$ to $\sqrt{c}\hspace{1pt} \ve_1$.
The standard construction of a Householder reflector matrix gives $\mmk{H}{c} = \mmk{I}{c} - \tfrac{\vu\vu^T}{c - \sqrt{c}}$ where $\vu = \sqrt{c}\hspace{1pt}\ve_1 - \vvk{e}{c}$.
Note that $\ve$ is orthogonal to the bottom $c-1$ rows of $\mmk{H}{c}$; since $\mmk{H}{c}$ is itself an orthogonal matrix, this implies the lower $c-1$ rows of $\mmk{H}{c}$ are an orthonormal basis for the nullspace of $\ve$.
Thus, an orthonormal basis for the eigenspace of $\mA_{K_c}$ with eigenvalue $-1$ is given by the set of vectors $\{n_j\}_{j=1}^{c}$, where
\[
    \vn_j = \ve_j + \left(\tfrac{1}{c - \sqrt{c}}\right)( \sqrt{c}\hspace{1pt}\ve_1 - \ve).
\]
By setting $\mnk{c} = [\vn_2, \cdots, \vn_c ]$, we have that
 $[ \tfrac{1}{\sqrt{c}} \vvk{e}{c},  \mnk{c} ]$ is a $c\times c$ orthogonal matrix and gives the useful identity
\begin{align}\label{eqn:N-property}
    \mnk{c}\mnk{c}^T &= \mmk{I}{c} - \tfrac{1}{c} \vvk{e}{c}\vvk{e}{c}^T.
\end{align}
Thus, $\mA_{K_c}$ has eigendecomposition as follows:
  $\lambda_1 = (c-1)$ with multiplicity 1 and eigenvector $\vv_1 = \tfrac{1}{\sqrt{c}} \ve$, and
  $\lambda_j = -1$ with multiplicity $(c-1)$ and eigenvector $\vv_j = \vn_j$.

\paragraph{Eigenvectors of $\mA_{G(c,m)}$}
Each eigenvector $\vv$ for $\mA_{G(c,m)}$ is a length $c(m+1)$ vector which we divide into a block of size $c$ and a block of size $cm$:
\[
    \mA_{G(c,m)} =
    \left(
        \begin{array}{cc}
            0 & \vvk{e}{m}^T \otimes \mIs{c} \\
            \vvk{e}{m} \otimes \mIs{c} & \mIs{m} \otimes \mA_{K_c}
        \end{array}
    \right), \quad
    \vv = \bmat{ \vvk{y}{c} \\ \vvk{w}{m} \otimes \vvk{x}{c} }.
\]
The length $c$ subvector $\vvk{y}{c}$ of each eigenvector corresponds to the independent set of $G(c,m)$; each $\vvk{x}{c}$ corresponds to one of the $K_c$ subgraphs.
We summarize the eigendecomposition of $\mA_{G(c,m)}$ in Table~\ref{tab:eigenvectors} and present an ordered list of the eigenvalues with their multiplicities in Table~\ref{tab:eigenvalues}.

\begin{table}
    \centering
    \footnotesize
    \begin{tabularx}{\textwidth}{lXr}
      eigenvalue & eigenvector & eigenspace dim\\
      \toprule
      \vspace{5pt}
      $\lambda = \left(\tfrac{1}{2}\right)\left( (c-1) \pm \sqrt{(c-1)^2 + 4m} \right)$ &   $\bmat{
                            \vvk{e}{c} \\
                            \tfrac{1}{\lambda - (c-1)} \vvk{e}{m} \otimes \vvk{e}{c}
                          }  \cdot  \left( \frac{(\lambda - \lambda_2)^2}{c ( (\lambda - \lambda_2)^2 + m )} \right)^{\tfrac{1}{2}}$
              & 1, each   \\
      \vspace{5pt}
      $\lambda = \left(\tfrac{1}{2}\right)\left( -1 \pm \sqrt{1 + 4m} \right)$  & $\bmat{
                \mnk{c} \\
                \tfrac{1}{\lambda + 1} \vvk{e}{m}\otimes\mnk{c}
              }  \cdot  \left( \frac{(\lambda + 1)^2}{(\lambda + 1)^2 + m} \right)^{\tfrac{1}{2}}$
              & $c-1$, each\\
      \vspace{5pt}
      $\lambda = (c-1)$ & $\bmat{ 0 \\ \mnk{m} \otimes \vvk{e}{c} }  \cdot  \left( \frac{1}{c} \right)^{\tfrac{1}{2}}$
      &  $m-1$ \\
      \vspace{5pt}
      $\lambda = -1$ & $\bmat{ 0 \\ \mnk{m} \otimes \mnk{c} } $ & $(c-1)(m-1)$ \\
    \end{tabularx}
    \caption{\label{tab:eigenvectors}
        Complete, orthonormal set of eigenvectors for $\mA_{G(c,m)}$. Each vector is length $c(m+1)$ and is divided into two subvectors: a top component of length $c$ whose entries correspond to the independent set of $G(c,m)$, and a bottom component of length $cm$ corresponding to the nodes in the cliques.
    }
\end{table}

\begin{table}
    \centering
    \begin{tabularx}{0.80\textwidth}{lrXr}
      eigenvalue & value & & multiplicity \\
      \toprule
      $\lambda_1$  &  $\left(\tfrac{1}{2}\right)\left( (c-1) + \sqrt{(c-1)^2 + 4m} \right)$ & & 1 \\
      $\lambda_2$  &  $c-1$  &  & $(m-1)$  \\
      $\lambda_3$  &  $\left(\tfrac{1}{2}\right)\left( -1 + \sqrt{1 + 4m} \right)$  &  & $(c-1)$  \\
      $\lambda_4$  &  $-1$  &  & $(m-1)(c-1)$  \\
      $\lambda_5$  &  $\left(\tfrac{1}{2}\right)\left( (c-1) - \sqrt{(c-1)^2 + 4m} \right)$  &  & 1  \\
      $\lambda_6$  &  $\left(\tfrac{1}{2}\right)\left( -1 - \sqrt{1 + 4m} \right)$  &  & $(c-1)$    \\
    \end{tabularx}
    \caption{\label{tab:eigenvalues}
        Ordered list of eigenvalues of $G(c,m)$ for $c, m \in \mathbb{N}^+$.
        The particular ordering
        $\lambda_1 > \lambda_2$ is guaranteed to hold because $c, m > 0$.
        The ordering $\lambda_2 > \lambda_3$ holds as long as $m < c^2 - c$,
        $\lambda_3 > \lambda_4$ always holds because $m > 0$,
        $\lambda_4 > \lambda_5$ holds as long as $c < m$,
        and finally $\lambda_5 > \lambda_6$ always holds because $c,m > 0$.
    }
\end{table}

\subsubsection*{Entropic sub-family of $G(c,m)$}

We now establish that an infinite sub-family of the graph class $G(c,m)$ is
entropic with respect to the matrix exponential and discuss its implications for
intervals of sub-entropic parameter values.
The proof of Theorem~\ref{thm:entropic-kks-general} is technical, and we defer it to the Appendix.

\begin{theorem}\label{thm:entropic-kks-general}
    There exists some $C \in \mathbb{N}^+$ such that for each $c \geq C$
there exists at least one value $\beta \in (0, \tfrac{1}{c-2})$ for which the two walk-classes of $G(c,c+1)$ have identical subgraph centrality.

\end{theorem}
\begin{corollary}\label{cor:entropic-kks-general}
  There exists some $C \in \mathbb{N}^+$ such that for each $c \geq C$
  there exists at least one value $\beta \in (0, \tfrac{1}{c-2})$ for which the graph $G(c,c+1)$ has maximum walk-entropy, i.e., $G(c,c+1)$ is entropic.
\end{corollary}

Since Corollary~\ref{cor:entropic-kks-general} exhibits a sequence of entropic values $\beta$ that converges to zero, we can rule out sub-entropic intervals near zero.
\begin{corollary}
    There is no $\eps > 0$ such that $(0,\eps)$ is a sub-entropic interval for subgraph centrality.
\end{corollary}

\section{Infinite families of $f$-entropic graphs}\label{sec:infinite-tensor}
In previous sections we considered graphs that are entropic with respect to the specific function $f(x) = \exp(x)$.
Here we show that there are infinitely many functions $f(x)$ for which $f$-entropic graphs exist.
More precisely, we show that for any analytic function that is entropic with respect to at least one graph, $G$, there is an infinite family of graph-function pairs such that the graph is entropic with respect to the function.
We construct these graphs using the graph tensor product $G \otimes H$, observing that it
satisfies $\mA_{G} \otimes \mA_{H} = \mA_{G \otimes H}$.

Given an $h$-entropic graph $G$ and walk-regular graph $H$, we will construct a function $f$ that is entropic on $G \otimes H$ under conditions described below.
We begin by proving some necessary lemmas.
\begin{lemma}\label{lem:walk-reg-triangle}
  Let $H$ be walk-regular and contain at least one triangle.
  Then $(H^k)_{jj} > 0$ for each $j$ and integers $k \geq 2$.
\end{lemma}
\begin{proof}
  Because $H$ contains a triangle, there exists a node $j \in H$ such that $(H^3)_{jj} > 0$.
  This implies $(H^3)_{ii} > 0$ for all nodes $i \in H$ by walk-regularity,
  and so each node is incident to at least one triangle.
  Thus, every node in $H$ has at least at one neighbor, and so $(H^2)_{ii} > 0$ for each $i$.

  Consider any integer $k\geq 4$.
  If $k$ is even, then each node $i$ must be incident to a closed $k$-walk: for example, take the walk from node $i$ to any of its neighbors and back to $i$, repeating until the length is $k$.
  If instead $k$ is odd, then it is at least 5. Consider the walk starting at $i$, traversing the triangle that must be incident to $i$ (as proved above), and then proceeding to a neighbor of $i$ and back to $i$ until the length is $k$.
  This proves that $(H^k)_{ii}$ is positive for each $i$ for $k \geq 4$.
\end{proof}

\begin{lemma}\label{lem:tensor-constant-diagonal-lemma}
  Given an $h$-entropic graph $G$ and a walk-regular graph $H$ that contains at least one triangle, we can construct a positive sequence $c_k > 0$ so that for the function $f(x) = \sum_{k=0}^{\infty} c_k x^k$,
   $f(\mA_{G \otimes H})$ is constant diagonal.
\end{lemma}
\begin{proof}
    To construct $f$ so that $f(\mA_{G \otimes H})$ is constant diagonal, consider an arbitrary diagonal entry $f( \mA_{G \otimes H})_{\ell \ell}$ for some $\ell$.
    Then $f(\mA_{G \otimes H})_{\ell \ell} =
    (\ve_i\otimes\ve_j)^Tf(\mA_{G\otimes H})(\ve_i\otimes\ve_j)$ for some $i$ and $j$.
    Using the fact $\mA_{G \otimes H} = \mA_{G} \otimes \mA_H$ and expanding the power series of $f$, we can write
    \begin{equation}
        f(\mA_{G \otimes H})_{\ell \ell}
        \hspace{3pt} = \hspace{3pt}
        \sum_{k=0}^{\infty} c_k  (\ve_i \otimes \ve_j)^T (\mA_{G} \otimes \mA_H)^k (\ve_i \otimes \ve_j)
        \hspace{3pt} = \hspace{3pt}
        \sum_{k=0}^{\infty} c_k (\mA_{G}^k)_{ii} (\mA_H^k)_{jj}. \label{eqn:tensor-decomp}
    \end{equation}
    By assumption, the graph $H$ is walk-regular, and so for each $k$ there is a constant $C_H(k)$ such that $(\mA_H^k)_{jj} = C_H(k)$ for all $j$.
    Moreover, since $H$ contains at least one triangle, $C_H(k)$ is positive for $k \geq 2$ by Lemma~\ref{lem:walk-reg-triangle}.
    Thus, we can choose $c_k$ as follows.
    By assumption $G$ is $h$-entropic, so we know there exists a PPSC function $h$ such that $h(\mA) = \sum_{k=0}^{\infty} h_k (\mA_G^k)$ is constant diagonal.
    Thus, we set $c_k = h_k / C_H(k)$ for $k\geq 2$, since $C_H(k)$ is positive there, and $c_0 = h_0, c_1 = h_1$.
    Substituting $c_k$ into Equation~\eqref{eqn:tensor-decomp} and simplifying, we get
    $f(\mA_{G\otimes H})_{\ell \ell} = \sum_{k=0}^{\infty} h_k (\mA_{G}^k)_{ii}$.

    The expression $\sum_{k=0}^{\infty}h_k \mA_G^k$ is constant diagonal by choice of the sequence $h_k$, and so for each $k$ there is a constant $C_G(k)$ such that $(\mA_{G}^k)_{ii} = C_G(k)$ for all $i$.
    Hence, for each $\ell$ we know $f(\mA_{G\otimes H})_{\ell \ell}$ equals the constant $\sum_{k=0}^{\infty} h_k C_G(k)$.
\end{proof}
\begin{theorem}\label{thm:tensor-entropic}
      Let $G$ be $h$-entropic, and let $H$ be walk-regular, connected, and contain at least one triangle.
      Then $G \otimes H$ is $f$-entropic for some PPSC function $f$.
\end{theorem}
\begin{proof}
    By Lemma~\ref{lem:tensor-constant-diagonal-lemma},
    we can construct a PPSC function $f$ so that $f(\mA_{G \otimes H})$ is constant-diagonal.
    Since $G$ is $h$-entropic, by definition it is not walk-regular,
    and so $G\otimes H$ is not walk-regular.
    To see this, note that $\mA_G^k$ must be not constant-diagonal for some power $k$, and so $(\mA_{G\otimes H})^k = (\mA_G^k)\otimes(\mA_{H}^k)$ is not constant-diagonal.
    Thus, to conclude that $G \otimes H$ is \emph{entropic} we need only show that $G \otimes H$ is connected.
    A result of Weichsel~(\cite{weichsel1962kronecker}, Theorem 1) states that a tensor graph $G_1 \otimes G_2$ is connected if and only if both $G_1$ and $G_2$ are connected and at least one of them contains a cycle of odd length.
    Since $H$ contains a triangle by assumption, we are done.
\end{proof}

For an entropic graph $G$, by induction we have that $G \otimes \left( \bigotimes_{j=1}^N H_j \right)$ is entropic for any set of connected, walk-regular graphs $H_j$ that each contain a triangle.
Moreover, letting $C_{\ell}$ denote the cycle graph on $\ell$ nodes,
we observe that for any connected, walk-regular graph $F$,
the graph $F \Box C_3$ is connected, walk-regular, and contains a triangle, where we again use
the result of Chiue and Shieh to prove connectedness.
To see that for any graph $H$, $H \Box C_3$ contains a triangle, note that the diagonal of $(\mA_{H \Box C_3})^3$ has positive entries.
Thus, $G \otimes (F \Box C_3)$ gives a distinct entropic graph for each connected, walk-regular graph $F$.
Finally, we note that since $G(4,5)$ is entropic and every cycle graph $C_k$ is connected and walk-regular, by Theorem~\ref{thm:tensor-entropic} each graph $G(4,5) \otimes (C_k \Box C_3)$ is $f_k$-entropic for some function $f_k$, yielding an infinite family of entropic functions $f_k$.

\section{Centrality collisions: when distinct walk-classes have identical centrality}\label{sec:constructing-entropic}
We want to understand when nodes with distinct walk structures can be assigned the same score by a centrality measure---we call this occurrence a \emph{centrality collision}, or simply a collision.
More precisely, let $G$ be a connected, non--walk-regular graph, and let $\{w_j\}$ be any collection of distinct walk-classes in $G$.
We want to know when we can construct a PPSC function $f(x)$ such that $f(\mA)_{ii}$ is the same for all nodes in the classes $\{w_j\}$;
we say that the walk-classes $\{w_j\}$ \emph{collide} under $f$, and that $f$ \emph{induces a collision} at $\{w_j\}$.
Observe that a graph is $f$-entropic precisely when there exists a function $f$ that induces a collision at all of its walk-classes.

In the remainder of the section, we give a sufficient condition for a graph for the existence of a collision-inducing PPSC function on that graph (Corollary~\ref{cor:pos-suff-condition});
the sufficient condition generalizes to apply to entropic graphs.
Additionally, this leads to a related sufficient condition for concluding that a graph is not $f$-entropic for any function $f$ (Corollary~\ref{cor:walk-class-lin-sys}).
Finally, as an application of our theory, in Section~\ref{sec:spider-torus} we present a graph with three walk-classes which we prove is $f$-entropic using Corollary~\ref{cor:pos-suff-condition}.
This is interesting because previously all known $f$-entropic graphs had only two walk-classes.

\subsection{Sufficient condition for centrality collisions}\label{sec:suff-condition}
Given a collection $\{w_j\}$ of distinct walk-classes in a graph $G$, we would like an easily computable condition that characterizes whether there exists some PPSC function $f$ that induces a collision at $\{w_j\}$.
Here we present a practically computable sufficient condition for such a function's existence.
Interestingly, the condition connects a question about nodes' walk-classes to Farkas's Lemma on nonnegative solutions to linear equations.

\begin{theorem}\label{thm:pos-suff-condition}
    Let graph $G$ have walk matrix $\mW$ and adjacency matrix $\mA$.
    Given $\vb > 0$, if there exists $\vx > 0$ such that $\mW\vx = \vb$, then there exists a PPSC function $f(x)$ such that $\diag(f(\mA)) = \vb$.
\end{theorem}
\begin{proof}
    Assume $\vx > 0$ and $\mW\vx =  \vb$.
    We will construct positive coefficients $\{c_k\}_{k=0}^{\infty}$ such that $f(t) = \sum c_k t^k$ is a convergent power series, and so $f$ is a PPSC function.

    Let $m$ be the degree of the minimal polynomial of $\mA$.
    For each index $k \geq m$ we can use the minimal polynomial of $\mA$ to produce a set of coefficients $\{ p_{k,j} \}_{j=0}^{m-1}$ such that
    $\mA^k = \sum_{j=0}^{m-1} p_{k,j} \mA^j$.

    To construct the sequence $c_k$, we begin by defining the terms $c_k$ for $k \geq m$.
    To do this, we first define a related sequence.
    For each $j=0, \cdots, m-1$, we want a sequence $\{s_{j,k}\}_{k=m}^{\infty}$ so that $\left(\sum_{k=m}^{\infty} p_{k,j} s_{j,k} \right)$ converges to some positive value.
    One such sequence is $s_{j,k} = 2^{-k} |p_{k,j}^{-1}|$;
    if $p_{k,j} = 0$ then instead we set $s_{j,k} = 2^{-k}$.
    Note that each $s_{j,k}$ is positive.
    Next, for each $k\geq m$, we set $c_k = \min\{  s_{j,k} | j=0,\cdots,m-1 \} > 0$.
    For each $j=0,\cdots, m-1$ this guarantees $|p_{k,j} c_k| \leq 2^{-k}$ for all $k$, and so by the limit comparison test
    $\left( \sum_{k=m}^{\infty} p_{k,j} c_k \right)$ is convergent.

    Finally, for $j=0, \cdots, m-1$, set $c_j = x_j - \left( \sum_{k=m}^{\infty} p_{k,j} c_k \right)$.
    If any of these $c_j$ is negative, then choose $\{ c_k \}_{k=m}^{\infty}$ smaller so that the values $c_j = x_j - \left( \sum_{k=m}^{\infty} p_{k,j} c_k \right)$ are positive.
    This is possible because the terms $c_k$ can be chosen as close to 0 as we like, and $x_j > 0$ by assumption.
    The end result is that any positive solution $\vx$ to $\mW\vx = \vb$ enables us to equate
    \[
        \vb \hspace{3pt} = \hspace{3pt} \mW \vx
         \hspace{3pt} =
        \hspace{3pt} \diag\left( \sum_{j=0}^{m-1} x_j \mA^j \right)
            = \hspace{3pt} \diag\left( \sum_{j=0}^{m-1} \left( c_j +  \sum_{k=m}^{\infty} p_{k,j} c_k \right) \mA^j \right).
    \]
    Rearranging, we have $\vb = \diag\left( \sum_{j=0}^{m-1} c_j \mA^j +  \sum_{j=0}^{m-1}\sum_{k=m}^{\infty} p_{k,j} c_k \mA^j \right)$.
    We observe
    \[
       \sum_{j=0}^{m-1} \left( \sum_{k=m}^{\infty} p_{k,j} c_k \right) \mA^j
     \hspace{3pt} = \hspace{3pt}
       \sum_{k=m}^{\infty} c_k \sum_{j=0}^{m-1} p_{k,j}  \mA^j
     \hspace{3pt} = \hspace{3pt}
        \sum_{k=m}^{\infty} c_k \mA^k,
    \]
    and thus $\vb$ equals $\diag(f(\mA))$ for a PPSC function $f$.
\end{proof}

Note that if we restrict to a subset $S$ of rows of the equation $\mW\vx = \ve$ corresponding to a subset of walk-classes,
then by Theorem~\ref{thm:pos-suff-condition} a positive solution $\vx$ to $\mW_S\vx = \ve$ gives a PPSC function $f$ that induces a collision at the walk-classes contained in $S$.

\begin{corollary}\label{cor:pos-suff-condition}
  Given a graph $G$ with walk matrix $\mW$ and adjacency matrix $\mA$, fix any subset $W = \{w_j\}$ of walk-classes for $G$ and let $J$ be the set of all row indices of $\mW$ corresponding to nodes in the walk-classes in $W$.
  If the linear system $\mW_J\vx = \ve$ has a solution $\vx > 0$, then there exists a PPSC function $f$ that induces a collision at all the walk-classes in $W$.
  That is, there is a constant $c$ such that $f(\mA)_{ii} = c$ for all $i \in w_j \in W$.
\end{corollary}

The above results show that a solution to a particular linear program guarantees that a collision-inducing function exists.
On the other hand, the converse---whether the existence of a collision-inducing PPSC function guarantees the existence of a positive solution to $\mW\vx=\ve$---remains an open question.
Lastly, we remark that the above proof shows that if an entropic function exists for $G$ then $\mW\vx = \ve$ is consistent.

\begin{corollary}\label{cor:walk-class-lin-sys}
  Given a graph $G$ with walk matrix $\mW$,
  if the linear system $\mW\vx = \ve$ is inconsistent then $G$ is non--walk-regular and not entropic.
\end{corollary}

\paragraph{Connection to Farkas's Lemma}
We remark that for the linear system in Corollary~\ref{cor:pos-suff-condition} to have a positive solution, it is necessary that the system $\mW\vx=\ve$ satisfies the well-known Farkas's Lemma~\cite{Farkas1902}.
Farkas's Lemma says that a general linear system $\mM\vx = \vb$ has a solution $\vx\geq 0$ if and only if no $\vy$ exists such that $\vy^T\vb < 0$ and $\vy^T\mM \geq0$.
In our restricted setting, where $\mW$ is nonnegative and the right-hand side is $\ve$, we derive a more specific, necessary condition for Farkas's Lemma to hold---and, therefore, a necessary condition for a positive solution to exist in Corollary~\ref{cor:pos-suff-condition}.
However, this novel condition does not have an intuitive interpretation in the context of centrality and walk-classes, and so we defer further discussion to the appendix.

\subsection{A concrete application}\label{sec:spider-torus}

Figure~\ref{fig:spidertorus} displays an $f$-entropic graph, $ST(4,2,[5,3])$, which we call a ``spider torus'', that has exactly three walk-classes.
The center nodes from the spider graphs form the first walk-class, the outer/inner nodes of the spider graphs (labelled with subscripts 1/2) form the second/third. We experimentally verified the graph $ST(4,2,[5,3])$ to be $f$-entropic using Corollary~\ref{cor:pos-suff-condition}, as implemented in our software package~\cite{spiderdonuts_v1.0.0}.

\begin{figure}[!ht]
  \centering
  \resizebox{0.9\linewidth}{!}{\input{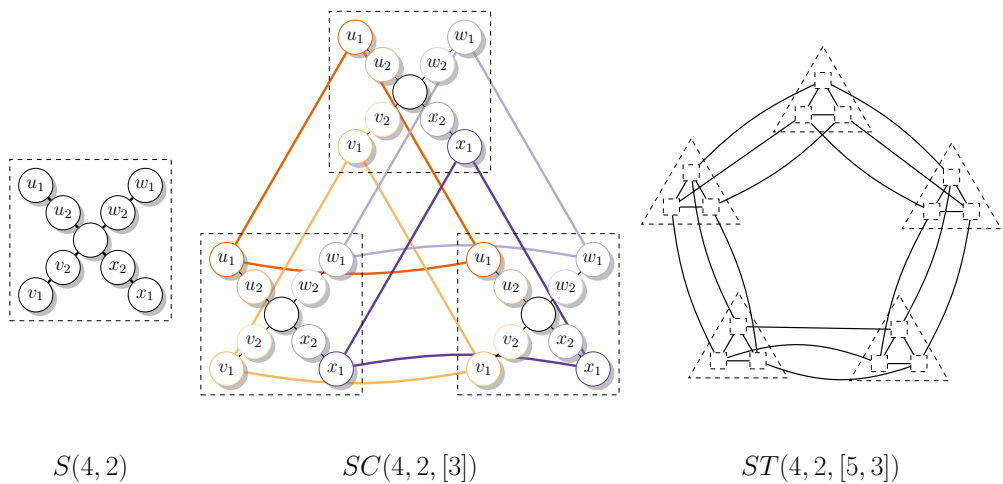}}
  \caption{
      \emph{(Left)} A spider graph of degree 4 and length 2, denoted $S(4,2)$.
      \emph{(Center)} A ``(4, 2, [3]) spider cycle'', denoted $SC(4,2,[3])$, consists of three copies of $S(4,2)$ such that the three copies of each outer-most node (with label $t_1$ for $t \in \{u,v,w,x\}$)  are connected in a cycle.
      \emph{(Right)} A ``(4, 2, [5, 3]) spider torus'', denoted $ST(4,2,[5,3])$, consists of five copies of $SC(4,2,[3])$ with cycles connecting each set of five inner nodes of the spider legs (with label $t_2$ for $t \in \{u,v,w,x\}$).  $ST(4,2,[5,3])$ has three walk-classes and can be shown to be $f$-entropic using Corollary~\ref{cor:pos-suff-condition}.
  }\label{fig:spidertorus}
\end{figure}

\section{Conclusions}\label{sec:conclusions}
We have considered when $f$-subgraph centrality measures induce collisions at different walk-classes in a graph; that is, when an $f$-subgraph centrality measure assigns identical scores to nodes that have different walk structures.
We settled two open questions about the cardinality of the set of entropic graphs, and the set of entropic values.
In particular, we exhibited an infinite family of graphs where subgraph centrality assigns identical scores to all nodes for two parameter values (Corollary~\ref{cor:cartesian-infinite});
furthermore, we constructed a separate infinite family of entropic graphs with entropic values $\beta_i$ that converge to zero (Corollary~\ref{cor:entropic-kks-general}), proving the set of all entropic values is infinite and has at least one limit point.
One consequence of this result is that no sub-entropic interval of the form $(0,\eps)$ exists.
It remains an open problem to determine whether there exists an interval $(a,b)$ that is sub-entropic for $\exp(\beta x)$.

The existence of graphs entropic with respect to $f(x) = \exp(x)$ raises the question of whether there exist $f$-entropic graphs for other functions $f$ commonly used to define centrality measures.
We resolve this question in the affirmative by exhibiting an infinite family of functions $f_i$ and graphs $G_i$ such that $G_i$ is $f_i$-entropic (Section~\ref{sec:infinite-tensor}).

Finally, we present conditions (Corollary~\ref{cor:pos-suff-condition} and Corollary~\ref{cor:walk-class-lin-sys}) that can prove the existence of a function $f$ inducing a collision at a set of walk-classes in a given graph $G$; that $G$ is $f$-entropic for some $f$; or that $G$ is sub-entropic for all $f$-subgraph centralities.
Each of these conditions can be efficiently evaluated in practice, and we use the first condition to exhibit an instance of an $f$-entropic graph with three walk-classes.
It remains an open question whether the sufficient condition in Corollary~\ref{cor:pos-suff-condition} is in fact a characterization of when nodes collide for some function $f$.

\section*{Acknowledgments}

This work supported in part by the Gordon $\&$ Betty Moore Foundations Data-Driven Discovery
Initiative through Grant GBMF4560 to Blair D. Sullivan.
The authors thank Michele Benzi for bringing walk-regular graphs to their attention.

\bibliographystyle{siamplain}
\bibliography{all-bibliography}

\appendix
%%%
%%%
%    APPENDIX
%%%
%%%

\section{Connecting walk-class collisions and Farkas's Lemma}
In Section~\ref{sec:suff-condition} we presented a sufficient condition on the walk matrix $\mW$ of a graph $G$ for concluding that $G$ is $f$-entropic, and we observed that for the sufficient condition to hold, the linear system $\mW\vx = \ve$ must satisfy Farkas's Lemma.
Because the system we consider, $\mW\vx = \ve$, has a specific structure, we are able to give a sharper necessary condition for Farkas's Lemma to hold in this setting (Definition~\ref{def:saff}).
We use $\avg(\vv)$ to denote the average of the entries of a nonnegative vector $\vv$.
\begin{definition}\label{def:saff}
    We say that a nonnegative matrix $\mM$ satisfies the \emph{set-average flip-flop property} (SAFF) if for every pair of disjoint, non-empty subsets $S$ and $T$ of row indices, there exists a column $j$ such that $\avg(\mM(T,j)) \leq \avg(\mM(S,j))$.
    We say a graph satisfies the set-average flip-flop property if its walk-matrix $\mW$ does.
\end{definition}
Equivalently, for each $S$, there must exist some $j$ so that $\avg( \mM(S,j) ) \geq \max\{  \mM(i,j) | i \notin S \}$.
\begin{lemma}
  Let $\mM$ be a nonnegative matrix and $\ve$ the vector of all 1s.
  Then for $\mM\vx = \ve$ to have a nonnegative solution, it is necessary that
  $\mM$ satisfies the set-average flip-flop property.
\end{lemma}
\begin{proof}
  Assume that $\mM$ does not satisfy SAFF. We will construct a vector $\vy$ such that $\vy^T\mM \geq 0$ but $\vy^T\ve < 0$; then, by Farkas's Lemma, there is no $\vx \geq 0$ such that $\mM\vx = \ve$.

  Since $\mM$ does not satisfy SAFF by assumption, there must exist disjoint, non-empty subsets $S,T$ such that for all columns $j$ we have $\avg(\mM(T,j)) \lneq \avg(\mM(S,j))$.
  Construct the vector $\vy$ as follows. Set $\vy(S) = 1/|S|$, and set $\vy(T) = -(1+\delta)/|T|$ for some $\delta > 0$ to be determined later in the proof.
  Then $\vy = \tfrac{1}{|S|}\ve_S - \tfrac{1+\delta}{|T|}\ve_T$, and we know
  $\vy^T\ve = -\delta$, so $\vy^T\ve < 0$
  as long as $\delta > 0$.

  Next we pick a specific $\delta > 0$ so that $\vy$ satisfies $\vy^T\mM \geq 0$.
  By construction, $\avg(\mM(T,j)) \lneq \avg(\mM(S,j))$ for all $j$, and by assumption $\mM$ is nonnegative, so the quantity
  \begin{equation}
    \delta = \displaystyle\min_i \left(  \frac{\avg(\mM(S,i)) - \avg(\mM(T,i))}{\avg(\mM(T,i))}   \right)
  \end{equation}
  is positive.
  With $\delta$ now defined, we will show $\vy^T\mM \geq 0$.
  Multiplying $\vy = \tfrac{1}{|S|}\ve_S - \tfrac{1+\delta}{|T|}\ve_T$ with column $j$ of $\mM$ gives
  \begin{align}
    \vy^T\mM(:,j) &= \avg(\mM(S,j)) - (1+\delta)\cdot \avg(\mM(S,j)) \nonumber \\
                  &= \avg(\mM(S,j)) - \avg(\mM(T,j)) - \delta \cdot \avg(\mM(T,j)). \label{eqn:affc-delta}
  \end{align}
  By construction of $\delta$, we know
  $\delta \leq ( \avg(\mM(S,j)) - \avg(\mM(T,j)) ) / \avg(\mM(T,j)) $,
  and so $-\delta \cdot \avg(\mM(T,j)) \geq - (\avg(\mM(S,j)) - \avg(\mM(T,j)))$.
  Substituting this in Equation~\eqref{eqn:affc-delta}, for each $j$ we have ${\vy^T\mM(:,j)\geq0}$.
  Thus, by Farkas's Lemma, no solution $\vx \geq 0 $ can exist to the equation $\mM\vx = \ve$.
\end{proof}

For Corollary~\ref{cor:pos-suff-condition} to imply that a graph $G$ is $f$-entropic, it is necessary that the walk matrix of $G$ satisfy the SAFF property.

%%%%
%%%%
%%    PROOFS
%%%%
%%%%
\section{Proofs from Section~\ref{sec:beta-distribution}}

\def\thetheorem{\ref{thm:entropic-kks-general}}
\begin{theorem}
    There exists some $C \in \mathbb{N}^+$ such that for each $c \geq C$
there exists at least one value $\beta \in (0, \tfrac{1}{c-2})$ for which the two walk-classes of $G(c,c+1)$ have identical subgraph centrality.

\end{theorem}

\begin{proof}[Proof of Theorem~\ref{thm:entropic-kks-general}]
      Given values $\beta, c, m$, let $\subscore{IS}{\beta}{c}{m}$ and $\subscore{CN}{\beta}{c}{m}$ denote the quantity $\exp(\beta \mA)_{jj}$ for independent set nodes and clique nodes in $G(c,m)$, respectively.
      We can explicitly produce expressions for $\subscore{IS}{\beta}{c}{m}$ and $\subscore{CN}{\beta}{c}{m}$ using the eigendecomposition in Table~\ref{tab:eigenvectors}.

      Observe that, because the independent set nodes have degree $m > c$, for small enough $\beta$ we know that $\subscore{IS}{\beta}{c}{m} > \subscore{CN}{\beta}{c}{m}$.
      Hence, to prove that an entropic value $\beta$ exists, by continuity it suffices to prove that $\subscore{CN}{\beta_0}{c}{m} > \subscore{IS}{\beta_0}{c}{m}$ for some $\beta_0 > 0$.
      To do this we will use the fact
      \[
          \exp(\beta\mA) = \sum_{k=1}^6 \exp(\beta \lambda_k) \mV_K \mV_k^T,
      \]
      where $\mV_k$ is composed of columns that form an eigen-basis for the eigenvalue $\lambda_k$ of $G(c,m)$.
      We remark that, because there are just two walk-classes in $G(c,m)$, the quantity $(\mV_k\mV_k^T)_{jj}$ can have only two values, depending on whether node $j$ is in a clique or the independent set.
      Replacing the quantities $(\mV_k\mV_k^T)_{jj}$ with expressions from Table~\ref{tab:eigenvectors} and Equation~\eqref{eqn:N-property} and simplifying we get the following:
      \begin{align*}
          \subscore{CN}{\beta}{c}{m} &=
                        e^{\beta\lambda_1} \tfrac{1}{c(\lambda_5^2 + m )} + e^{\beta\lambda_5} \tfrac{1}{c(\lambda_1^2 + m )}
                        + (1- \tfrac{1}{c})\left(e^{\beta\lambda_3} \tfrac{1}{(\lambda_3+1)^2 + m}
                        + e^{\beta\lambda_6} \tfrac{1}{(\lambda_6+1)^2 + m} \right),\\
                        &~~+ e^{\beta\lambda_2}\tfrac{1}{c}(1 - \tfrac{1}{m})
                      + e^{\beta\lambda_4}(1 - \tfrac{1}{c})(1 - \tfrac{1}{m}) \\
          \subscore{IS}{\beta}{c}{m} &=
              e^{\beta\lambda_1} \tfrac{\lambda_5^2 }{c(\lambda_5^2 + m)}
              + e^{\beta\lambda_5} \tfrac{\lambda_1^2 }{c(\lambda_1^2 + m)}
              + (1- \tfrac{1}{c})\left( e^{\beta\lambda_3} \tfrac{(\lambda_3+1)^2}{(\lambda_3+1)^2 + m}
              + e^{\beta\lambda_6} \tfrac{(\lambda_6+1)^2}{(\lambda_6+1)^2 + m} \right).
      \end{align*}

    Our goal is to prove $\subscore{CN}{\beta_c}{c}{m} > \subscore{IS}{\beta_c}{c}{m}$ for some $\beta_c$ for all $c$ larger than some threshold $C$.
    We proceed by splitting the above functions into pieces, which we bound independently.
    Let
    \begin{align*}
        h_1(c,m) &=  \tfrac{1}{c}\left(e^{\beta\lambda_1} \tfrac{1}{(\lambda_5^2 + m )}
                            + e^{\beta\lambda_5} \tfrac{1}{(\lambda_1^2 + m )}\right)
                  + \left(e^{\beta\lambda_2}-e^{\beta\lambda_4} - (e - 2) \right)\tfrac{1}{c}(1 - \tfrac{1}{m}), \\
        h_2(c,m) &= \left( e^{\beta \lambda_4} + \tfrac{e-2}{c}\right)(1-\tfrac{1}{m}), \\
        g_1(c,m) &= \tfrac{1}{c}\left( e^{\beta\lambda_1} \tfrac{\lambda_5^2}{\lambda_5^2 + m} + e^{\beta\lambda_5}\tfrac{\lambda_1^2}{\lambda_1^2+m} \right), \text{and} \\
        g_2(c,m) &= (1 - \tfrac{1}{c}) \left( e^{\beta\lambda_3} \tfrac{(\lambda_3+1)^2-1}{(\lambda_3+1)^2+m} + e^{\beta \lambda_6} \tfrac{(\lambda_6+1)^2-1}{(\lambda_6+1)^2+m} \right).
    \end{align*}
    Then $\subscore{CN}{\beta}{c}{m} = h_1(c,m) + h_2(c,m)$ and
    $\subscore{IS}{\beta}{c}{m} = g_1(c,m) + g_2(c,m)$, and it suffices to show
    $h_1(c,m) > g_1(c,m)$ and $h_2(c,m) > g_2(c,m)$.
    We proceed by handling these inequalities separately. We remark that, although we assume $m = c+1$ throughout, we continue to write things in terms of $c$ and $m$ for clarity.

\paragraph{Proving that $h_1 > g_1$.}
    First, note that $h_1(c,m) > g_1(c,m)$ holds if and only if $c\cdot h_1(c,m) > c \cdot g_1(c,m)$.
    Second, subtracting the smallest exponential terms, $e^{\beta\lambda_1} \tfrac{1}{(\lambda_5^2 + m )}$ and $e^{\beta\lambda_5} \tfrac{1}{(\lambda_1^2 + m )}$, from both sides, it thus suffices to show
    \begin{equation}\label{eqn:main-piece1}
      \left(e^{\beta\lambda_2}-e^{\beta\lambda_4} - (e - 2) \right)(1 - \tfrac{1}{m})
    > e^{\beta\lambda_1} \tfrac{\lambda_5^2 - 1}{\lambda_5^2 + m} + e^{\beta\lambda_5}\tfrac{\lambda_1^2 - 1}{\lambda_1^2+m}.
    \end{equation}

    Recall that the exponential satisfies $1- x < e^{-x} < 1 - \tfrac{x}{1+x}$.
    Since $\beta\lambda_2 = 1+\tfrac{1}{c-2}$ and $\lambda_4 = -1$, this implies
    $e^{\beta\lambda_2} > e(1 + \tfrac{1}{c-2})$
    and
    $-e^{\beta\lambda_4} > -(1 - \tfrac{1}{c-1}).$
    Thus, we can write
    \[
        \left(e^{\beta\lambda_2}-e^{\beta\lambda_4} - (e - 2) \right)(1 - \tfrac{1}{m}) > (\tfrac{e+1}{c-1} + 1)(1 - \tfrac{1}{m}),
    \]
    and Inequality~\eqref{eqn:main-piece1} follows if
    \begin{equation}\label{eqn:nolambdaleft}
      (\tfrac{e+1}{c-1} + 1)(1 - \tfrac{1}{m}) > e^{\beta\lambda_1} \tfrac{\lambda_5^2 - 1}{\lambda_5^2 + m} + e^{\beta\lambda_5}\tfrac{\lambda_1^2 - 1}{\lambda_1^2+m}.
    \end{equation}

    Now we upperbound the right-hand side, handling
    each term separately. Since $\lambda_5 < -1$, $e^{\beta \lambda_5} < e^{-\tfrac{1}{c-2}}$,
    we know $e^{\beta \lambda_5} < 1 - \tfrac{1}{c-1} < 1 - \tfrac{1}{m}$.
    Using the fact that $\tfrac{\lambda_1^2 - 1}{\lambda_1^2 + m} < 1$, we can write
    \[
        e^{\beta \lambda_5} \tfrac{\lambda_1^2 - 1}{\lambda_1^2 + m} < (1 - \tfrac{1}{m}).
    \]
    Next, since $m = c+1$, we know that $\lambda_1 < c+2$, therefore
    $e^{\beta \lambda_1} < e^{ 1 + \tfrac{3}{c-2}}$.

    Standard algebraic manipulation shows $\lambda_5^2 < 2$ holds,
    which implies $\lambda_5^2 < 2 m/(m-1)$.
    This allows us to show $\tfrac{\lambda_5^2-1}{\lambda_5^2 + m} < \tfrac{1}{m}$,
    and substituting above yields the inequality
    \[
        e^{\beta \lambda_1} \tfrac{\lambda_5^2-1}{\lambda_5^2 + m} < e^{1 + \tfrac{3}{c-2}}\tfrac{1}{m}.
    \]
    We can now replace Inequality~\eqref{eqn:nolambdaleft} with
    \[
        (\tfrac{e+1}{c-1} + 1)(1 - \tfrac{1}{m})
        > e^{1 + \tfrac{3}{c-2}}\tfrac{1}{m} + (1 - \tfrac{1}{m}).
    \]
    Subtracting $(1-\tfrac{1}{m})$ from both sides, we need to show
    \[
        \tfrac{e+1}{c-1}(1 - \tfrac{1}{m})
        >  e^{1 + \tfrac{3}{c-2}}\tfrac{1}{m}.
    \]
    Multiplying both sides by $c-1$, noting that $m = c+1$, and taking the limit as $c\rightarrow \infty$
    yields $e+1$ on the left and $e$ on the right, completing the proof that $h_1 > g_1$.

    \paragraph{Proving that $h_2 > g_2$.}
    Since $(1-1/m) > (1-1/c)$, it suffices to show
    \begin{equation}\label{eqn:g2h2}
      \left(e^{\beta \lambda_4} + \tfrac{e-2}{c}\right)
      >
      e^{\beta\lambda_3} \tfrac{(\lambda_3+1)^2-1}{(\lambda_3+1)^2+m} + e^{\beta \lambda_6} \tfrac{(\lambda_6+1)^2-1}{(\lambda_6+1)^2+m}.
    \end{equation}
    We begin by simplifying the fractions containing $\lambda_3$ and $\lambda_6$ and transforming the right-hand side into hyperbolic trig expressions.

    Setting $\gamma = \sqrt{4m+1}$, substitution and algebra yields the following identities
    \begin{align}
            \tfrac{(\lambda_3+1)^2}{(\lambda_3+1)^2 + m} & = \tfrac{1}{2}(1+\tfrac{1}{\gamma})
            & \textrm{and} \qquad\qquad
            \tfrac{(\lambda_6+1)^2}{(\lambda_6+1)^2 + m} & = \tfrac{1}{2}(1-\tfrac{1}{\gamma}), \label{eqn:lambda3gamma}\\
            \tfrac{1}{(\lambda_3+1)^2 + m} & = \tfrac{1}{2}\tfrac{1}{m}(1-\tfrac{1}{\gamma})
            & \textrm{and} \qquad\qquad
            \tfrac{1}{(\lambda_6+1)^2 + m} & = \tfrac{1}{2}\tfrac{1}{m}(1+\tfrac{1}{\gamma}).\label{eqn:lambda6gamma}
    \end{align}

    Substituting Equations~\eqref{eqn:lambda3gamma} and~\eqref{eqn:lambda6gamma}
    into the right-hand side of Inequality~\eqref{eqn:g2h2} and rearranging we get that $e^{\beta\lambda_3} \tfrac{(\lambda_3+1)^2-1}{(\lambda_3+1)^2+m} + e^{\beta \lambda_6} \tfrac{(\lambda_6+1)^2-1}{(\lambda_6+1)^2+m}$ is
    \begin{align*}
        =~& e^{\beta\lambda_3} \left(\tfrac{1}{2}(1+\tfrac{1}{\gamma}) - \tfrac{1}{2m}(1-\tfrac{1}{\gamma}) \right)
         + e^{\beta \lambda_6} \left(\tfrac{1}{2}(1-\tfrac{1}{\gamma}) - \tfrac{1}{2m}(1+\tfrac{1}{\gamma}) \right) \\
        =~&
        \tfrac{1}{2}\left(
            (e^{\beta\lambda_3} + e^{\beta\lambda_6})(1-\tfrac{1}{m})
            +
            (e^{\beta\lambda_3} - e^{\beta\lambda_6})\tfrac{1}{\gamma}(1+\tfrac{1}{m})
        \right).
    \end{align*}
    To transform this expression into hyperbolic trig functions,
    first observe that
    \begin{equation}\label{eqn:lambda36beta}
      \beta \lambda_3 = - \tfrac{1}{2} \beta + \tfrac{1}{2}\beta \gamma
      \quad\quad \textrm{and} \quad\quad
      \beta \lambda_6 = - \tfrac{1}{2} \beta - \tfrac{1}{2}\beta \gamma.
    \end{equation}
    Setting $\xi = \tfrac{1}{2}\beta \gamma$, we can write
    \begin{align*}
      e^{\beta\lambda_3} \tfrac{(\lambda_3+1)^2-1}{(\lambda_3+1)^2+m} + e^{\beta \lambda_6} \tfrac{(\lambda_6+1)^2-1}{(\lambda_6+1)^2+m}
      &=
      \tfrac{1}{2}e^{-\tfrac{1}{2}\beta}\left(
          (e^{\xi} + e^{-\xi})(1-\tfrac{1}{m})
          +
          (e^{\xi} - e^{-\xi})\tfrac{1}{\gamma}(1+\tfrac{1}{m}) \right) \\
      &=
         e^{-\tfrac{1}{2}\beta} \left( \cosh(\xi)(1-\tfrac{1}{m}) + \tfrac{1}{\gamma} \sinh(\xi) (1+\tfrac{1}{m}) \right).
    \end{align*}

    Thus, to show Inequality~\eqref{eqn:g2h2} it suffices to prove
    \[
        e^{-\tfrac{1}{2}\beta} + e^{\tfrac{1}{2}\beta}(\tfrac{e-2}{c})
        > \cosh(\xi)(1-\tfrac{1}{m}) + \tfrac{1}{\gamma} \sinh(\xi) (1+\tfrac{1}{m}).
    \]

    Using the standard inequality $(1+x) \leq e^{x}$, we have $e^{-\tfrac{1}{2}\beta} > (1 - \tfrac{1}{2}\beta)$ and $e^{\tfrac{1}{2}\beta} >  (1 + \tfrac{1}{2}\beta)$, so it suffices to prove that
    \[
        (1 - \tfrac{1}{2}\beta) + (1 + \tfrac{1}{2}\beta) (\tfrac{e-2}{c})
        > \cosh(\xi)(1-\tfrac{1}{m}) + \tfrac{1}{\gamma} \sinh(\xi) (1+\tfrac{1}{m}).
    \]
    We accomplish this by splitting into two inequalities as follows
    \begin{align}
        (1 - \tfrac{1}{2}\beta) + \tfrac{2}{\gamma}(1+\tfrac{1}{2}\beta)(\tfrac{e-2}{c})
        &> \cosh(\xi)(1-\tfrac{1}{m}), \label{eqn:cosh-piece} \\
        (1-\tfrac{2}{\gamma})(1+\tfrac{1}{2}\beta)(\tfrac{e-2}{c})
        &> \tfrac{1}{\gamma}\sinh(\xi)(1+\tfrac{1}{m})\label{eqn:sinhineq}.
    \end{align}
    We begin by showing~\eqref{eqn:sinhineq}.
    Multiplying by $c$ and rearranging, we have
    \[
    (1-\tfrac{2}{\gamma})(1+\tfrac{1}{2}\beta)(e-2)
    > \frac{\sinh\left(\tfrac{\sqrt{c+5/4}}{c-2}\right)}{ 2 \tfrac{\sqrt{c+5/4}}{c} } \cdot (1+\tfrac{1}{m}).
    \]
    Letting $c\rightarrow \infty$ yields $e-2$ on the left and $1/2$ on the right, since $\lim\limits_{x\rightarrow 0} \tfrac{\sinh(x)}{x} = 1$.

    It remains to show Inequality~\eqref{eqn:cosh-piece}. Recall that
    $m =  c+1$,
    $\beta = \tfrac{1}{c-2}$,
    $\gamma = \sqrt{1 + 4m} = 2 \sqrt{c + 5/4}$,
        and $\xi = \tfrac{1}{2} \beta \gamma = \tfrac{\sqrt{c+5/4}}{c-2}$.
    Then Inequality~\eqref{eqn:cosh-piece} holds if and only if
    \[
         (e-2)(1+\tfrac{1}{2}\beta)
        >
        c\cdot\gamma\cdot \left( \cosh(\xi)(1-\tfrac{1}{m}) - (1 - \tfrac{1}{2}\beta) \right).
    \]
    Taking the limit as $c\rightarrow \infty$, the left-hand side goes to $(e-2)$.
    We will show the right-hand side converges to 0.
    To see this, we rewrite the right-hand side as
    \[
        \tfrac{c}{c-2} \gamma \left( (c-2)(\cosh(\xi)-1)  - \tfrac{1}{2} \right)
        - \gamma(\cosh(\xi) - 1) + \tfrac{2\gamma}{c-2} + \tfrac{\gamma}{m} \cosh(\xi) - \tfrac{c\gamma}{(c-2)m}.
    \]
  As $c$ increases, $\xi \rightarrow 0$ and $\cosh(\xi) \rightarrow 1$;
  thus, $\tfrac{\gamma}{m} \cosh(\xi)$ vanishes, and $\gamma(\cosh(\xi) - 1)$ vanishes by the well-known trigonometric fact $\lim\limits_{x\rightarrow 0} \tfrac{\cosh(x)-1}{x} = 0$.
  The terms $\tfrac{2\gamma}{c-2}$ and $\tfrac{c\gamma}{(c-2)m}$ vanish by the power rule.

  All that remains is to show that
  $\gamma \left( (c-2)(\cosh(\xi)-1)  - \tfrac{1}{2} \right)$ converges to 0 as $c$ increases.
  Recall that $\gamma = \sqrt{4m+1} = 2(c+5/4)^{1/2}$.
  From the power series expansion of $\cosh(\xi)$, we have
 $(c+5/4)^{1/2} \left( (c-2)(\cosh(\xi) -1 ) - \tfrac{1}{2} \right)$ is
  \begin{align*}
      =~& (c+5/4)^{1/2} \left(  (c-2)\cdot\left( \sum_{k=1}^{\infty} \tfrac{1}{(2k)!} \left( \tfrac{\sqrt{c+5/4}}{c-2} \right)^{2k}  \right) - \tfrac{1}{2} \right) \\
     =~& (c+5/4)^{1/2} \left(  \left( \sum_{k=2}^{\infty} \tfrac{1}{(2k)!}
     \tfrac{(c+5/4)^{k} }{(c-2)^{2k-1}} \right) + \tfrac{1}{2}\tfrac{c+5/4}{c-2} - \tfrac{1}{2} \right) \\
     =~& \left( \sum_{k=2}^{\infty} \tfrac{1}{(2k)!}
     \tfrac{(c+5/4)^{k + 0.5} }{(c-2)^{2k-1}}\right) + \tfrac{1}{2} (c+5/4)^{1/2} \left( \tfrac{c+5/4}{c-2} - 1\right).
  \end{align*}
  For $c$ large enough, the summation can be bounded above, term for term, by the geometric series
  $\sum_{k=1} (\sqrt{c})^{-k} = \frac{c^{-1/2}}{1 - c^{-1/2}},$
  which vanishes as $c$ increases.
  Finally, observe
  \[
  \tfrac{1}{2}(c+5/4)^{1/2}\left(\tfrac{c+5/4}{c-2} - 1\right) = \tfrac{1}{2}(c+5/4)^{1/2}\cdot \tfrac{13}{4(c-2)},
  \]
  which also converges to 0 as $c$ increases.
  Thus, $h_2 > g_2$, which completes the proof.

\end{proof}

\end{document}